\documentclass{eptcs}
 \providecommand{\firstpage}{49}

\usepackage{breakurl}

\usepackage{amssymb,amsmath,amsthm}
\usepackage{graphicx}
\usepackage{xcolor}
\usepackage{url}

\newtheorem{theorem}{Theorem}
\newtheorem{lemma}[theorem]{Lemma}
\newtheorem{definition}[theorem]{Definition}
\newtheorem{example}[theorem]{Example}

\newcommand{\Sigmaast}{\Sigma^\ast}
\newcommand{\Sigmanod}{\Sigma_{\mkern1mu 0}}
\newcommand{\Sigmanodast}{\Sigma^\ast_{\mkern1mu 0}}

\newcommand{\mbR}{\mathbin{\R}}
\newcommand{\mkzero}{\mkern1mu 0 \mkern1mu}
\newcommand{\lc}{\lbrace \:}
\newcommand{\qprime}{{q \mkern1mu}'}
\newcommand{\rc}{\: \rbrace}
\newcommand{\singleton}[1]{\lbrace {#1} \rbrace}
\newcommand{\size}[1]{\#{#1}}
\newcommand{\sort}{\textsl{sort}}
\newcommand{\toG}{\Rightarrow}
\newcommand{\toGbar}{\Rightarrow}
\newcommand{\idr}{\Rightarrow_{i/d}}
\newcommand{\idrast}{\Rightarrow_{i/d}^{\ast}}
\newcommand{\toGast}{\Rightarrow^\ast}
\newcommand{\toa}{\stackrel{a}{\rightarrow}}
\newcommand{\toaone}{\xrightarrow{a_1}}

\newcommand{\toas}{\xrightarrow{a_s}}

\newcommand{\togkpk}{\Rightarrow_{g_k,p_k}}
\newcommand{\togp}{\Rightarrow_{g,p}}
\newcommand{\togrpr}{\Rightarrow_{g_r,p_r}}

\newcommand{\toprimeb}{\mathrel{\stackrel{b}{\rightarrow}\mkern-3mu{}'}}
\newcommand{\toprimebone}{\mathrel{\hbox{$\xrightarrow{b_1}\mkern-3mu{}'$}}}

\newcommand{\toprimebs}{\mathrel{\hbox{$\xrightarrow{b_s}\mkern-3mu{}'$}}}
\newcommand{\toprimeepsilon}{\mathrel{\xrightarrow{\varepsilon}\mkern-3mu{}'}}

\newcommand{\toprime}{\rightarrow'}

\newcommand{\dom}{\mathit{dom}}

\newcommand{\rng}{\mathit{rng}}
\newcommand{\slice}{\mathit{s\ell}}
\newcommand{\yield}{\mathit{yield}}

\newcommand{\G}{G}
\renewcommand{\L}{\mathcal{L}}
\newcommand{\R}{\sim}
\newcommand{\X}{\mathcal{X}}

\newcommand{\blankline}{\vspace{0.8\baselineskip}}

\newcounter{alphaitemcnt}

\title{Combining Insertion and Deletion in RNA-editing \\
  Preserves Regularity} 

\author{%
  E.P. de Vink${}^\ast$
  \institute{%
    Department of Mathematics and Computer Science,
    Technische Universiteit Eindhoven 
    \smallskip \\
    Centrum Wiskunde en Informatica,
    Amsterdam
  }
  \and 
  H. Zantema
  \institute{%
    Department of Mathematics and Computer Science,
    Technische Universiteit Eindhoven 
    \smallskip \\
    Institute for Computing and Information Sciences,
    Radboud University Nijmegen}
  \and 
  D. Bo\v{s}na\v{c}ki
  \institute{%
    Department of Biomedical Engineering,
    Technische Universiteit Eindhoven
  }
}

\def\titlerunning{Combining Insertion and Deletion in RNA-editing 
  Preserves Regularity}
\def\authorrunning{E.P.~de Vink, H.~Zantema and D.~Bo\v{s}na\v{c}ki}

\begin{document}

\maketitle

%% \pagestyle{plain}
%% \pagenumbering{arabic}
%% \setcounter{page}{1}

\begin{abstract}
  \textbf{Abstract} Inspired by {RNA}-editing as occurs in
  transcriptional processes in the living cell, we introduce an
  abstract notion of string adjustment, called guided rewriting. This
  formalism allows simultaneously inserting and deleting elements.  We
  prove that guided rewriting preserves regularity: for every regular
  language its closure under guided rewriting is regular too. This
  contrasts an earlier abstraction of {RNA}-editing separating
  insertion and deletion for which it was proved that regularity is
  not preserved. The particular automaton construction here relies on
  an auxiliary notion of slice sequence which enables to sweep from
  left to right through a completed rewrite sequence.
\end{abstract}

\def\thefootnote{\fnsymbol{footnote}}
\footnotetext[1]{Corresponding author, \texttt{evink@win.tue.nl}}

\section{Introduction}

We study an elementary biologically inspired formalism of string
replacement referred to as guided rewriting. Given a fixed and finite
set~$G$ of strings, also called guides, a rewriting step amounts to
adapting a substring towards a guide. We consider two versions of
guided rewriting: \emph{guided insertion/deletion}, which is close
to an editing mechanism as encountered in the living cell, and general
guided rewriting based on an \emph{adjustment relation}, which is
mathematically more amenable. For guided insertion/deletion the guide
and the part of the string that is rewritten do not need to be of the
same length. They are required to be equal up to occurrences of a
distinguished dummy symbol.  For general guided rewriting the
correspondence of the guide and the substring that is rewritten is
element-wise. The guide and substring are equivalent symbol-by-symbol
according to a fixed equivalence relation called adjustment.

In both cases, for a finite set of guides~$G$, only a finite set of
strings can be obtained by repeatedly rewriting a given
string. Starting from a language~$L$, we may consider the
extension~$L_{i/d}$ of the language with all the rewrites obtained by
guided insertion/deletion and the extension~$L_G$ of the language
obtained by adding all the adjustment-based guided rewrites.  We
address the question if regularity of~$L$ implies regularity
of~$L_{i/d}$ and of~$L_G$.  The results of the paper state that in the
case of guided insertion/deletion regularity is preserved if the
strings of dummy symbols involved are bounded and that guided
rewriting based on adjustment always preserves regularity.

The motivation for this work stems from transcriptional biology. RNA
can be seen as strings over the alphabet $\{ C, G, A, U
\}$. Replication of the encoded information is one of the most
essential mechanisms in life: strands of RNA are faithfully copied by
the well-known processes of RNA-transcription. However, typical for
eukaryotic cells, the synthesis of RNA does not yield an exact copy of
part of the DNA, but a modification obtained by post-processing. The
class of the underlying adjustment mechanisms is collectively called
\emph{RNA-editing}.

Abstracting away from biological details, the
computational power of insertion-deletion systems for RNA-editing is
studied in~\cite{TY03}: an insertion step is the replacement of a
string~$u v$ by the string $u \alpha v$ taken from a particular finite
set of triples $u, \alpha, v$. Similarly, a deletion step replaces $u
\alpha v$ by~$u v$ for another finite set of triples $u, \alpha
,v$. In~\cite{MPRV05} the restriction is considered where $u$ and~$v$
are both empty.  The approach claims full computational power, that
is, they generate all recursively enumerable languages.

In the RNA-editing mechanisms occurring in nature, however, only very
limited instances of these formats apply. Often only the symbol~$U$ is
inserted and deleted, instead of arbitrary strings~$\alpha$, see
e.g.~\cite{ATS97}. Therefore, following~\cite{Zan10:lata}, we
investigate guided insertion/deletion focusing on the special role of
the distinguished symbol~$0$, the counterpart of the
RNA-base~$U$. However, in order to prove that under this scheme
regularity is preserved we extend our investigations to guided
rewriting based on an abstraction adjustment relation. In fact, we
prove the theorem for guided insertion/deletion by appealing to the
result for guided rewriting based on adjustment.

The proof of the latter result relies on reorganizing sequences of
guided rewrites into sequences of so-called slices. The point is that,
since guides may overlap, each guided rewrite step adds a `layer' on
top of the previous string. In this sense guided rewriting is
vertically oriented. E.g., Figure~\ref{fig-five-slices} in
Section~\ref{sectr} shows six rewrite steps of the string
$\mathit{ebcfa}$ yielding the string $\mathit{fbcfb}$ involving eight
layers in total. However, in reasoning about recognition by a finite
automaton a horizontal orientation is more natural. One would like to
sweep from left to right, so to speak. Again referring to
Figure~\ref{fig-five-slices}, five slices can be distinguished,
viz.\ a slice for each symbol of the string $\mathit{ebcfa}$. The
technical machinery developed in this paper allows for a transition
between the two orientations.

%% Reflecting this restriction, in~\cite{Zan10:lata} a
%% form of guided rewriting is studied in which only instances of one
%% distinguished symbol~$0$, the counterpart of the RNA-base~$U$, are
%% either added or removed, the number being independent of the
%% surrounding strings $u$ and~$v$. Moreover, guides are given by a fixed
%% and finite set of guides~$G$, just like in RNA-editing itself. In this
%% setting only strings consisting of the symbol~$0$ are inserted or
%% deleted. One result of~\cite{Zan10:lata} states that this kind of
%% guided rewriting does not preserve regularity. The result strongly
%% depends on the strict separation of insertion and deletion steps. In
%% the current paper we propose a variant of this setting in which mixing
%% of insertion and deletion in one step is allowed and prove that, in
%% contrast to the case of its restricted form, regularity is preserved.

%% We lift this to a more general approach: instead of considering $0$ as
%% a single symbol, we introduce symbols $0^k$ representing groups of $k$
%% consecutive $0$'s, for a limited number of $k$ values. This gives rise
%% to an equivalence relation $\R$: all symbols $0^k$ are equivalent to
%% each other, and none of the other symbols are equivalent. In this way
%% the setting of insertion and deletion of $0$'s is an instance of the
%% more general setting based on an arbitrary equivalence relation on
%% symbols for which we prove our main theorem.

The organization of this paper is as follows. The biological
background of RNA-editing is provided in Section~\ref{mrnaediting}.
Section~\ref{sec-ins-del} presents the theorem on preservation of
regularity for guided insertion-deletion. The notion of guided
rewriting based on an adjustment relation is introduced in
Section~\ref{sec-guided-rewriting} and a corresponding theorem on
preservation of regularity is presented. To pave the way for its
proof, Section~\ref{sectr} introduces the notions of a rewrite
sequence and of a slice sequence and establishes their
relationship. Rewrite sequences record the subsequent guided rewrites
that take place, slice sequences represent the cumulative effect of
all rewrites at a particular position of the string being adjusted. In
Section~\ref{sec-preservation} we provide, given a finite automaton
accepting a language~$L$, the construction of an automaton for the
extended language~$L_G$ with respect to a set of guides~$G$.
Section~\ref{seccr} wraps up with related work and concluding remarks.

\blankline

\paragraph{Acknowledgment}
We acknowledge fruitful feedback from Peter van der Gulik and detailed
comment from the reviewers of the MeCBIC 2012 workshop.

\clearpage

\section{Biological motivation}
\label{mrnaediting}

This section provides a description of RNA editing from a biological
perspective. In this paper we focus on the insertion and deletion of
uracil in messenger RNA (mRNA) and provide abstractions of the
underlying mechanism in the sequel. However, in the living cell there
are different kinds of RNA editing that vary in the type of RNA that
is edited and the type of editing operations.  Uracil is represented
by the letter~$U$. The three other types of nucleotides for RNA, viz.\
adenine, guanine and cytosine are represented by the letters $A$, $G$
and~$C$, respectively.

$U$-insertion/deletion editing is widely studied in the mitochondrial
genes of kinetoplastid protozoa~\cite{StuartEtAl97}. Kinetoplastids
are single cell organisms that include parasites
like~\emph{Trypanosoma brucei} and~\emph{Cri\-thi\-dia fasciculata}, that
can cause serious diseases in humans and/or animals. Modifications of
kinetoplastid mRNA are usually made within the coding regions. These
are the parts that are translated into proteins, which are the
building blocks of the cells. This way coded information of the
original gene can be altered and therefore expressed, i.e.\ translated
into proteins, in a varying number of ways, depending on the
environment in the cell. This provides additional flexibility as well
as potential specialization of different parts of the organisms for
particular functions.

\blankline

\noindent
Here we describe a somewhat simplified version of the mechanism for
the insertion and deletion of~$U$. More details can be found, for
instance, in~\cite{StuartEtAl97,ATS97,Blum1990189,vdSpekEtAl91}.  For
simplicity we assume that only identical letters match with one
another. In reality, the matching is based on complementarity, usually
assuming the so-called Crick-Watson pairs: $A$~matches with~$U$ and
$G$~matches with~$C$.

In general, a single step in the editing of mRNA involves two strands
of RNA, a strand of messenger RNA and a strand of guide RNA, the
latter typically referred to as the guide. To explain the mechanism
for the insertion of uracil, let us consider an example. See
Figure~\ref{fig-example}. Assume that we start of with an mRNA
fragment: $u = N_1N_2N_3N_4N_5$ and the guide $g = N_2N_3UUUN_4$,
where $N_i$ can be an arbitrary nucleotide $A$, $G$ or~$C$, but not
$U$. Obviously, there is some match between $u$ and~$g$ involving the
letters $N_2$, $N_3$, and~$N_4$, which is partially `spoiled' by the
$UUU$ sequence. By pairing of letters we have that $g$ attaches
to~$u$; the matching substrings $N_2 \mkern1mu N_3$ and~$N_4$ serve 
as anchors.

\begin{figure}[hbtp]
  \centering
  \includegraphics[scale=0.65]{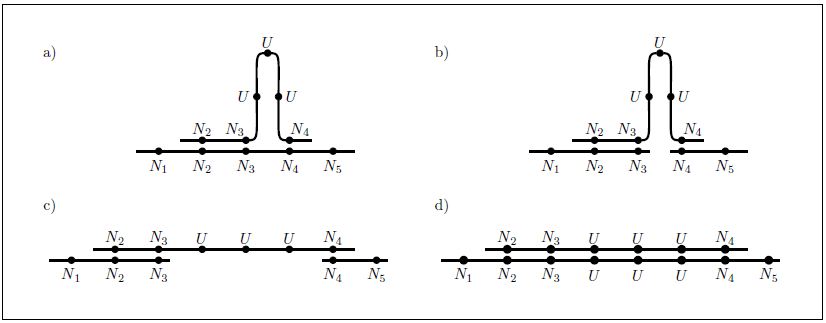}
  \caption{Various stages of guided $U$-insertion}
  \label{fig-example}
\end{figure}

By chemical reactions involving special enzymes $u$ is split open
between $N_3$ and~$N_4$.  The gap between the anchors is then filled
by the enzyme mechanism using the guide as a template. For each
letter~$U$ in the guide a~$U$ is added also in the gap.  As a result
the mRNA string~$u$ is transformed into $N_1N_2N_3UUUN_4N_5$. In
general, one can have more than two anchors (involving only non-$U$
letters) in which the guide and the mRNA strand match. In that case
mRNA is opened between each pair of anchors and all gaps between these
anchors are filled with~$U$ such that the number of $U$s in the guide
is matched.

A similar biochemical mechanism implements the deletion of~$U$s from a
strand of mRNA\@. We illustrate the deletion process on the following
example. Let us assume that we have the mRNA strand $u =
N_1N_2N_3UUN_4N_5$ and the guide $g = N_2N_3N_4$. Like in the
insertion case, $g$~initiates the editing by attaching itself to~$u$
at the matching positions $N_2$,$N_3$, and~$N_4$. Only now the
enzymatic complex removes the mismatching $UU$ substring between $N_3$
and~$N_4$ to ensure the perfect match between the substring and the
guide.  As a result the edited string $N_1N_2N_3N_4N_5$ is
obtained. In general, we can have several anchoring positions on the
same string. In that case, all~$U$s between each two matching
positions are removed from the mRNA\@.

A guide can also induce both insertions and deletions of~$U$
simultaneously. For instance the guide $N_2N_3UUUN_4$ can induce
editing in parallel of the string $N_1UN_2UN_3UN_4UN_5UN_6$ which
results in the string $N_1UN_2N_3UUUN_4UN_5UN_6$, where the~$U$
between $N_2$ and~$N_3$ has been deleted and two~$U$'s between $N_3$
and~$N_4$ have been inserted. This is done by the same biochemical
mechanisms that are involved in separate insertions and
deletions. Analogously as above, we can have multiple insertions and
deletions induced by the same guide on the original pre-edited
sequence.

The net effect of all three cases considered above is that a strand~$u
= xyz$, such that $y$ equals~$g$ up to occurrences of~$U$, is modified
by the insertion and deletion mechanism and becomes a string $v =
xgz$.  It is noteworthy that the rewriting system that we describe in
the sequel also applies to another case with the same effect.  For
example, consider a guide $g = N_2N_3UUUN_4$ and a pre-edited mRNA $u
= N_1N_2N_3UUN_4N_5N_6$.  Now, to obtain the match of the guide~$g$
and a substring~$y$ of~$u$, a~$U$ is inserted in~$u$, resulting in the
string $v = N_1N_2N_3UUUN_4N_5N_6$. If the $U$~subsequence in~$y$ was
longer though, like in the case for $u' = N_1N_2N_3UUUN_4N_5N_6$ and
$g' = N_2N_3UUN_4$, then we have that the extra~$U$ in~$u'$ is removed
resulting in $v' = N_1N_2N_3UUN_4N_5N_6$.

To summarize, the mRNA editing mechanism underlying
$U$-insertion/deletion can be interpreted as symbolic manipulations of
strings. In the sequel symbol~$U$ will be denoted by~$0$ and obviously
plays a special role. The crucial point is that in a single step some
substring~$y$ is replaced by a guide~$g$ for which $y$ and~$g$
coincide except for the symbol~$0$.

%%  This mechanism is the basis of guided insertion/deletion rewriting
%% as presented in the next section.

\section{Guided insertion\,/\,deletion}
\label{sec-ins-del}

Inspired by the biological scheme of editing of mRNA as discussed in
the previous section, we study the more abstract notion of guided
insertion and deletion and guided rewriting based on an adjustment
relation in the remainder of this paper. In this section we address
guided insertion and deletion, turning to guided rewriting in
Section~\ref{sec-guided-rewriting}.

More precisely, fix an alphabet~$\Sigmanod$ and distinguish $0 \notin
\Sigmanod$. Put $\Sigma = \Sigmanod \cup \singleton{0}$. Choose a
finite set $\G \subseteq \Sigmaast$, with elements~$g$ also referred
to as guides. Reflecting the biological mechanism, we assume that each
$g \in \G$ has at least two letters and that the first and last letter
of each $g \in \G$ are not equal to~$0$. Hence, $G \subseteq \Sigmanod
{\cdot} \Sigma^\ast {\cdot} \Sigmanod$. Now a guided
insertion/deletion step $\idr$ with respect to~$\G$ is given by
\begin{displaymath}
  {u \idr v} \iff 
  {{u = xyz} \: \land \: 
   {v = xgz} \: \land \: 
   {g \in \G} \: \land \: 
   {\pi(y) = \pi(g)}}
\end{displaymath}
where $y \in \Sigmanod {\cdot} \Sigma^\ast {\cdot} \Sigmanod$, and
$\pi(y)$ and~$\pi(g)$ are obtained from $y$ and~$g$, respectively, by
removing their~$0$s. Thus, $\pi : \Sigmaast \to \Sigmanodast$ is the
homomorphism such that $\pi(\varepsilon) = \varepsilon$, $\pi(0) =
\varepsilon$ and $\pi(a) = a$ for $a \in \Sigmanod$. So, intuitively,
$g$~is anchored on the substring~$y$ of~$u$ and sequences of~$0$s are
adjusted as prescribed by the guide~$g$, in effect replacing the
substring~$y$ by the guide~$g$ while maintaining the prefix~$x$ and
suffix~$z$.

As a simple example of a single guided insertion/deletion step, for
$\G = \singleton{g}$ with $g = bcb000ab0c$ and $u =
a00bc00babcc00a00b$ we have $u \idr v$ for $v =
a00bcb000ab0cc00a00b$. Here it holds that $u = a00 \cdot bc00babc
\cdot c00a00b$, $\pi ( bc00babc ) = bcbabc = \pi ( bcb000ab0c )$ and
$v = a00 \cdot bcb000ab0c \cdot c00a00b$. Note, for the string~$v$,
being the result of a rewrite with guide~$g$ itself with only one
possible anchoring, only trivial steps can be taken further. So, the
operation of guided insertion/deletion with the same guide~$g$ at the
same position in a string is idempotent. However, anchoring may
overlap. Consider the set of guides $\G = \lc aa0a ,\, a0aa \rc$, for
example. Then the string $aaa$ yields an infinite rewrite sequence
\begin{displaymath}
  aaa \idr aa \mkzero a \idr a \mkzero aa \idr aa \mkzero a \idr a
  \mkzero aa \cdots {}
\end{displaymath}
Still, from~$aaa$ only finitely many different rewrites can be
obtained by insertion/deletion steps guided by this~$\G$, viz.\ $\lc
aaa ,\, aa0a ,\, a0aa \rc$.

The restrictions put on~$G$ exclude arbitrary deletions (possible if
$\varepsilon$ would be allowed as guide) and infinite pumping (if
guides need not be delimited by symbols from~$\Sigmanod$). As an
illustration of the latter case, starting from the string $abc$ and
`guide' $0ab$, the infinite sequence $\mathit{abc} \idr \mkzero abc
\idr \mkzero \mkzero abc \idr \mkzero \mkzero \mkzero abc \ldots {}$
would be obtained. The restriction on the substring~$y$ prevents to
make changes outside the scope of the guide~$g$ and forbids $a \mkzero
b \mkzero \mkzero \mkzero c \idr ab \mkzero c$ by way of the guide
$ab$.

As a first observation we show that the set $L^u_{i/d} = \lc v \in
\Sigmaast \mid u \idrast v \rc$, for any finite set of guides~$\G$ and
any string~$u$, is finite. Write $u = a_0 0^{\mkern1mu i_1}
a_1 \ldots a_{n_1} 0^{\mkern1mu i_n} a_n$ where $a_i \in \Sigmanod$,
$i_k \geqslant 0$, for some $n \geqslant 0$. In effect, a guided
insertion/deletion step only modifies the substrings~$0^{\mkern1mu
  i_k}$ or leaves them as is. Therefore, after one or more guided
insertion/deletion steps the substrings~$0^{\mkern1mu i_k}$ are
strings taken from the set
\begin{displaymath}
  Z_{i/d}^{\mkern1mu u} =
  \lc 0^{\mkern1mu i_k} \mid 1 \leqslant k \leqslant n \rc \cup 
  \lc 0^{\mkern1mu \ell} \mid xa \cdot 0^{\mkern1mu \ell} bz \in
  \G ,\ a,b \in \Sigmanod ,\ \ell \geqslant 0 \rc
\end{displaymath}
Thus, if $u \idrast v$ then $v \in \hat{L}^u_{i/d}$, where
$\hat{L}^u_{i/d} = \lc a_0 \mkern1mu z_1 a_1 \ldots a_{n_1} \mkern1mu
z_n \mkern1mu a_n \mid z_k \in Z_{i/d}^{\mkern1mu u},\ 1 \leqslant k
\leqslant n \rc$, i.e.\ $L^u_{i/d} \subseteq \hat{L}^u_{i/d}$. Since
the set~$G$ is finite, it follows that $Z_{i/d}^{\mkern1mu u}$ is
finite, that $\hat{L}^u_{i/d}$ is finite and that $L^u_{i/d}$ is
finite as well.

More generally, given a set of guides~$G$, we define the extension by
insertion/deletion~$L_{i/d}$ of a language~$L$ over~$\Sigma$ by
putting $L_{i/d} = \lc v \in \Sigma^\ast \mid \exists u \in L \colon u
\idr^* v \rc$. Casted to the biological setting of
Section~\ref{mrnaediting}, $L$~are the strands of messenger RNA, $G$
are strands of guide RNA\@. Next, we consider the question whether
regularity of the language~$L$ is inherited by the induced
language~$L_{i/d}$. Note, despite the finiteness of the
insertion/deletion scheme for a single string, it is not obvious that
such would hold. 

For example, consider the language corresponding to the regular
expression $(ab)^\ast$ together with the operation $\sort$ which maps
a string~$w$ over the alphabet $\singleton{a,b}$ to the string
$a^{\mkern1mu n} b^{\mkern1mu m}$ where $n = \#_a(w)$, $m =
\#_b(w)$. Thus $\sort \mkern1mu (w)$ is a sorted version of~$w$ with
the $a$'s preceding the~$b$'s. Note, for $w \in (ab)^\ast$ there is
only one string $\sort \mkern1mu (w)$, as sorting is a deterministic,
hence finitary operation. However, despite $\L( \, (ab)^\ast \,)$, the
language associated by the regular expression, is regular, the
language
\begin{displaymath}
  \sort \mkern1mu ( \, (ab)^\ast \, ) = 
  \lc \sort \mkern1mu (w) \mid w \in (ab)^\ast \rc = 
  \lc a^{\mkern1mu n} b^{\mkern1mu n} \mid n \geqslant 0 \rc
\end{displaymath}
is \emph{not} regular. Also, if we define the rewrite operation $ba
\to_R ab$, then $\lc v \in \singleton{a,b}^\ast \mid u \to_R^\ast v
\rc$ contains shuffles of the string~$u$, i.e.\ all strings over
$\singleton{a,b}$ having the same number of $a$'s and $b$'s but are
smaller lexicographically. Thus, the set $\lc v \in
\singleton{a,b}^\ast \mid u \to_R^\ast v \rc$ is finite for each
string~$u$. However, the language $\hat{L} = \lc v \in
\singleton{a,b}^\ast \mid \exists u \in L \colon \ u \to_R^\ast v \rc$
cannot be regular: intersection with the language of $a^\ast b^\ast$
does not yield a regular language. More specifically, $\hat{L} \cap
\L( \, a^\ast b^\ast \,) = \lc a^{\mkern1mu n} b^{\mkern1mu n} \mid n
\geqslant 0 \rc$. We conclude that the question of $L_{i/d}$ being
regular, given regularity of the language~$L$, is not straightforward.

With the machinery of rewrite sequences and slice sequences developed
in the sequel of the paper, we will be able to prove the following for
guided insertion/deletion.

\blankline

\begin{theorem}
  \label{thmid}
  In the setting above, if~$L$ is a regular language and for some
  number~$k \geqslant 0$ it holds that no string of~$L$ or~$\G$
  contains $k$ (or more) consecutive~$0$'s, then the language
  $L_{i/d}$ is regular too.
\end{theorem}

\blankline

% In~\cite{Zan10:lata} a variant of guided rewriting was introduced
% based on a special symbol~$0$, and in which in every step $0$'s are
% either inserted or deleted. More precisely, a finite set~$G$ of guides
% is given, as well as an adjustment, being the replacement of a
% string~$s$ by a string~$g \in \G$, in which $g$ is obtained from~$s$
% by either inserting or deleting copies of the symbol~$0$. So, by
% insert steps the length of the string increases, while by delete steps
% the length decreases. A contribution from~\cite{Zan10:lata} is that
% for this type of guided rewriting regularity is \emph{not} preserved.

% A natural variant of this insertion/deletion, inspired by biological
% RNA-editing and in line with our notion of guided rewriting, allows a
% mixture of insertions and deletions. More precisely, a step is the
% replacement of a string~$s$ by a string~$g \in G$ for which $g$ can
% be obtained from~$s$ by inserting or deleting copies of the special
% symbol~$0$, which corresponds to U in the setting of RNA-editing.  The
% difference with the variant from~\cite{Zan10:lata} is that now in a
% single step insertion and deletion may be combined. For instance, if
% $ab0c \in G$, then $a00bc$ can be replaced by $ab0c$. In this section
% we apply Theorem~\ref{thmmain} to conclude that for this combined
% variant regularity is preserved, for languages with a restricted
% number of consecutive~$0$s.  First we give some definitions and
% assumptions.

% Due to the requirement that guides do not start or end in~$0$, also 
% $L_{i/d}$ will not contain strings with $k$ consecutive~$0$'s.

\noindent
We will prove Theorem~\ref{thmid} by applying a more general result on
guided rewriting, viz.\ Theorem~\ref{thmmain} formulated in the next
section and ultimately proven in Section~\ref{sec-preservation}. As in
the notion of guided rewriting as developed in the sequel, symbols are
only replaced by single symbols by which lengths of strings are always
preserved, a transformation is required to be able to apply
Theorem~\ref{thmmain}.

Before doing so we relate our results to those of
\cite{Zan10:lata}. There a relation similar to $\idr$ was introduced,
with the only difference that in a single step either $0$'s are
deleted or inserted, but not simultaneously. One of the conclusions
of~\cite{Zan10:lata} is that in that setting regularity is {\em not}
preserved, so the opposite of the main result in the present setting.

\section{Guided rewriting}
\label{sec-guided-rewriting}

The idea of guided rewriting is that symbols are replaced by
equivalent symbols with respect to some {\em adjustment
  relation}~$\R$.  The one-one correspondence of the symbols of the
string~$u$ and its guided rewrite~$v$, enjoyed by this notion of
reduction, will turn out technically convenient in the sequel.

%% In the setting of insertion-deletion this will be the relation on
%% $\overline{\Sigma}$ by which $0_i \R 0_j$ for all $0 \leq i,j < k$, and
%% other symbols are not related by $\R$. First we give some definitions.

Let $\Sigma$ be a finite alphabet and~$\R$ an equivalence relation
on~$\Sigma$, called the {\em adjustment relation}. If $a \mbR b$ we
say that $a$ can be adjusted to~$b$. For a string~$u \in \Sigmaast$ we
write $\#u$ for its length, use $u[i]$ to denote its $i$-th element,
$i = 1,\ldots,\# u$, and let $u[p,q]$ stand for the substring $u[p]
\, u[p{+}1] \cdots u[q]$. The relation~$\R$ is lifted to
$\Sigmaast$ by putting
\begin{displaymath}  
  u \R v \textit{\quad iff \quad} \#u = \#v \: \land \: \forall i = 1, \ldots,
  \#u \colon u[i] \R v[i]
\end{displaymath}

\noindent
Next we define a notion of guided rewriting that involves an
adjustment relation.

\begin{definition}
  \label{df-guided-rewriting}
  We fix a finite subset $\G \subseteq \Sigmaast \!$, called the set of
  guides.
  \begin{itemize}
  \item [(a)] For $u,v \in \Sigmaast \!$, $g \in \G$, $p \geqslant 0$, we
    define $u \togp v$, stating that $v$~is the rewrite of~$u$ with
    guide~$g$ at position~$p$, by
    \begin{displaymath}
      u \togp v \quad \mathit{iff} \quad \exists x, y, z \in
      \Sigmaast \colon u = x \, y \, z 
      \ \land \ \#x = p \ \land \ y \R g \ \land \  v = x \, g \, z
    \end{displaymath}
  \item [(b)] We write $u \toG v$ if $u \togp v$ for some $g \in \G$
    and $p \geqslant 0$. We use $\toGast$ to denote the reflexive
    transitive closure of~$\toG$. A sequence $u_1 \toG u_2 \toG \cdots
    \toG u_n$ is called a reduction.
  \item [(c)] For a language~$L$ over $\Sigma$ and a set of
    guides~$\G$ we write
    \begin{displaymath}
      L_\G = \lc v \in \Sigma^\ast \mid \exists u \in L \colon u
      \toGast v \rc
    \end{displaymath}
  \end{itemize}
\end{definition}

\noindent
So, a $\toG$-step adjusts a substring to a guide in~$\G$ element-wise,
and $L_\G$ consists of all strings that can be obtained from a string
from $L$ by any number of such adjustments.  For example, if $\Sigma =
\lbrace a, b, c \rbrace$, $\G = \{bb\}$ and $a \R b$ but not $a \R c$,
then by a $\toG$-step two consecutive symbols not equal to~$c$ are
replaced by two consecutive~$b$'s. In particular, $aaacaa \to_{bb,1}
abbcaa$ and $abbcaa  \to_{bb,0} bbbcaa$. We have

\begin{displaymath}
  \{aaacaa\}_\G = \lc  aaacaa ,\, bbacaa,\, abbcaa,\, aaacbb,\,
  bbbcaa,\, abbcbb,\, bbacbb, bbbcbb \rc
\end{displaymath}

\blankline

\noindent
Next, we state the main result of this paper regarding guided
rewriting as given by Definition~\ref{df-guided-rewriting}.

\begin{theorem}
  \label{thmmain}
  Given an equivalence relation~$\R$ on~$\Sigma$, let $G$ be a finite
  set of guides. Suppose $L$ is a regular language. Then $L_\G$ is
  regular too.
\end{theorem}

\noindent
Before going to the proof, we first show that both finiteness of~$G$
and the requirement of~$\R$ being an equivalence relation are
essential. Below, for a regular expression~$r$ we write $\L(r)$ for
its corresponding language.

To see that finiteness of~$\G$ is essential for Theorem~\ref{thmmain}
to hold, let $G = \lc c \, a^{\mkern2mu k} c \, b^{\mkern2mu k} c \mid
k \geqslant 0 \rc$ and $L = \L(c \, a^*c \, a^* c)$. Let $\R$ satisfy
$a \R b$ but not $a \R c$.  Then all elements of $L$ on which an
adjustment is applicable are of the shape $ca^{\mkern2mu
  k}ca^{\mkern2mu k}c$, where the result of the adjustment is
$ca^{\mkern2mu k}cb^{\mkern2mu k}c$, which can not be changed by any
further adjustment. So
\begin{displaymath}
  L_\G \; \cap \; \L(c \, a^* c \, b^* c) 
  \; = \; 
  \lc c \, a^{\mkern1mu k} c \mkern1mu\,  b^{\mkern2mu k} c \mid k
  \geqslant 0 \rc 
\end{displaymath}
is not regular. Since regularity is closed under intersection we
conclude that $L_G$ cannot be regular itself.

Also equivalence properties of~$\R$ are essential for Theorem
\ref{thmmain}. For $\G = \lc ab \rc$ and ${\R} = \lc (a,b), (b,a)
\rc$ the only possible $\toG$-steps are replacing the pattern~$ba$
by~$ab$. Note that here $\R$ is neither reflexive nor transitive.
Since $b a$ may be replaced by~$a b$, bubble sort on $a$'s and $b
\mkern1mu$'s can be mimicked by $\toGast$, while on the other hand
$\toGast$ preserves both the number of~$a$'s and the number
of~$b$'s. Hence
\begin{displaymath}
  \L((a b)^*)_\G \; \cap \; \L(a^* b^*) \; = \; \lc a^{\mkern2mu k}
  b^{\mkern2mu k} \mid k \geqslant 0 \rc
\end{displaymath}
which proves that $\L((ab)^*)_\G$ is not regular, again since
regularity is closed under intersection.

%% \begin{example}
%%   Let $\A = \lc 1, 2, 3 \rc$. Put $g = 1203$, $h = 2301$ and
%%   take $\G = \lc g , h \rc$ as set of guides. Then for $u =
%%   1002003001$ we have both $u \toGast 12301$ and $u \toGast 120301$:
%%   \begin{displaymath}
%%     \begin{array}{lclcl}
%%       1002003001 & \to_{g,0} & 1203001 & \to_{h,1} & 12301  \\
%%       1002003001 & \to_{h,1} & 1002301 & \to_{g,0} & 120301 
%%     \end{array}
%%   \end{displaymath}%
%%   We see that the order of the steps influences the result. Starting
%%   from $12301$ an infinite sequence of rewrite steps is possible by
%%   continuing $12301 \to_{g,0} 120301 \to_{h,1} 12301$. We have
%%   \begin{displaymath}
%%     \lc 1002003001 \rc_\G = \lc 1002003001, \, 1203001, \, 1002301,
%%     \, 12301, \, 120301 \rc
%%   \end{displaymath}
%% \end{example}

\section{Rewrite sequences and slice sequences}
\label{sectr}

Fix an alphabet~$\Sigma$, an adjustment relation~$\R$, and a set of
guides~$\G$.

\begin{definition}
  A sequence $\varrho = ( g_k, p_k )^r_{k=1}$ of guide-position pairs
  is called a guided rewrite sequence for a string~$u \in \Sigmaast$
  if it holds that (i)~$g_k \in \G$, (ii)~$0 \leqslant p_k \leqslant
  \#u - \#g_k$, and (iii)~$u[ p_k{+}1, p_k{+}\#g_k ] \mbR g_k$, for
  all~$k = 1, \ldots, r$.
\end{definition}

\noindent
A guide-position pair~$(g,p)$ indicates a redex for a guided rewrite
with~$g$ of the string~$u$. The position~$p$ is relative to~$u$. For
the rewrite to fit we must have $p + \#g \leqslant \#u$. The first
$p$~symbols of~$u$, i.e.\ the substring $u[1,p]$, are not affected by
the rewrite, as are the last $\#u-p+\#g$ symbols of~$u$, i.e.\ the
substring $u[ p{+}\#g{+}1 , \#u ]$.

The sequence~$\varrho$ induces a sequence of strings~$( u_k )^r_{k=0}$
by putting $u_0 = u$ and $u_{k}$ such that $u_{k{-}1} \togkpk u_{k}$
for $k = 1, \ldots, r$. To conclude that $u_{k{-}1} \togkpk u_{k}$ is
indeed a proper guided rewrite step, in particular that we have
$u_{k-1} [p_k{+}1, p_k{+}\#g_k ]$, we use the assumption $u[
  p_k{+}1,p_k{+}\#g_k ] \mbR g_k$ and the fact that if $u \togp v$
then $u [p+1,p+\#g] \mbR v [ p+1, p+\#g ]$. So we obtain $u \toGast
u_r$ by construction.  The string~$u_r$ is referred to as the yield
of~$\varrho$ for~$u$, notation $\yield(\varrho)$.  Conversely, every
specific reduction from~$u$ to~$v$ gives rise to a corresponding
guided rewrite sequence for~$u$.

%% A guided rewrite sequence~$\varrho$ is called repetition-free, if we
%% have that $g_k = g_\ell \land p_k = p_\ell$ only if $k = \ell$.

\begin{definition}
  Let $a \in \Sigma$.  A sequence $\slice = ( g_i , q_i )_{i \in I}$
  of guide-offset pairs, for $I \subseteq \mathbb{N}$ a finite index
  set, is called a slice for~$a$ and~$\G$ if it holds that (i)~$g_i \in
  \G$, (ii)~$1 \leqslant q_i \leqslant \#g_i$, and~(iii) $a \mbR
  g_i[q_i]$, for all $i \in I$. The slice~$\slice$ is called a slice
  for a string~$u \in \Sigmaast$ at position~$n$, $1 \leqslant n
  \leqslant \#u$, if it is a slice of~$u[n]$.
\end{definition}

\noindent
Note that in a guide-offset pair~$(g,q)$ of a slice sequence, the
offset~$q$ is relative to the guide~$g$. Since we require $1 \leqslant
q \leqslant \#g$ for such a pair, the symbol~$g[q]$ is
well-defined. We will reserve the use of~$q$ for offsets, indices
within a guide, and the use of~$p$ for positions after which a rewrite
may take place, i.e.\ for lengths of proper prefixes of a given
string.

The goal of the notion of slice is to summarize the effect of a number
of guided rewrites local to a specific position within a string. The
symbol generated by the last rewrite that affected the position,
i.e.\ the particular symbol of the last element of the slice sequence,
is part of the overall outcome of the total rewrite. This symbol is
called the {\em yield} of the slice. More precisely, if $I \neq
\emptyset$, the yield of a slice~$\slice$ for a symbol~$a$ is defined
as $\yield( \slice ) = g_{i_{\max}}[ q_{i_{\max}} ]$ where $i_{\max} =
\max(I)$. In case $I = \emptyset$, we put $\yield( \slice ) =
a$. Occasionally we write $a \mbR \slice$, as for a slice~$\slice$ for
a symbol~$a$ it always holds that $a \mbR \yield( \slice )$.

A slice~$\slice$ is said to be repetition-free if $g_i = g_j \land
q_i = q_j$ implies $i = j$. If we have $I = \emptyset$, the
slice~$\slice$ is called the empty slice.

\blankline

\noindent
Next we consider sequences of slices, and investigate the relationship
between slices on two consecutive positions in a guided rewrite
sequence.

\begin{definition}
  A sequence $\sigma = ( \slice_n )^{\#u}_{n=1}$ is called a slice
  sequence for a string~$u$ if the following holds:
  \begin{itemize}
  \item $\slice_n$ is a slice for~$u$ at position~$n$, for $n = 1 ,
    \ldots , \#u$;
  \item for $n = 1, \ldots, \#u{-}1$, putting $\slice_n = ( g_i ,
    q_i )_{i \in I}$ and $\slice_{n+1} = (g'_i , q'_i \mkern1mu )_{i \in J}$,
    there exists a monotone partial injection $\gamma_n : I \to J$
    such that, for all $i \in I$ and $j \in J$,
    \begin{itemize}
    \item $i \notin \dom(\gamma_n) \implies \mkern2mu q_i = \#g_i$
    \item \hspace*{0.1cm} $\gamma_n (i) = j \iff g_i = g'_j \land
      q_i+1 = q'_j$
    \item $j \notin \rng(\gamma_n) \mkern4mu \implies q'_j = 1$
    \end{itemize}
  \item the slices $\slice_1$ and $\slice_{\#u}$, say $\slice_1 = (
    g_i , q_i )_{i \in I}$ and $\slice_{\#u} = ( g'_j , q'_j )_{j \in
      J}$, satisfy $q_i = 1$, for all $i \in I$, and $q'_j = \#g'_j$,
    for all $j \in J$, respectively.
  \end{itemize}
\end{definition}

\noindent
For the slices $\slice_n$ and~$\slice_{n+1}$ the mapping $\gamma_n
\colon I \to J$ is called the cut for $\slice_n$
and~$\slice_{n{+}1}$. It witnesses that $\slice_n$
and~$\slice_{n{+}1}$ match in the sense that a rewrite may end at
position~$n$, may continue for its next offset at position~$n{+}1$,
and may start at position~$n{+}1$.  Since a cut~$\gamma$ is an
order-preserving bijection from $\dom( \gamma )$ to~$\rng( \gamma )$,
and $\dom( \gamma )$ and~$\rng( \gamma )$ are finite, it follows that
for two slices $\slice, \slice'$ the cut $\slice \to \slice'$ is
unique. We write $\slice \leadsto \slice'$.  A slice $\slice = ( g_i,
q_i )_{i \in I}$ is called a start slice if $q_i = 1$ for all~$i \in
I$. Similarly, $\slice$ is called an end slice if $q_i = \#g_i$ for
all~$i \in I$. A start slice is generally associated with the first
position of the string that is rewritten, an end slice with the last
position. Note, a start slice as well as an end slice are allowed to
be empty. The yield of the slice sequence~$\sigma$ is the sequence of
the yield of its slices, i.e.\ we define $\yield(\sigma) = \yield(
\slice_1 ) \cdots \yield( \slice_{\#u} )$.

\blankline

\begin{example}
  \label{ex-slice-sequence}
Let~$\R$ be the adjustment relation with equivalence classes
$\{a,b\},\{c,d\},\{e,f\}$ and let the set of guides~$G$ be given by
$\G = \lc g_1, \, g_2, \, g_3 \rc$ where $g_1 = \mathit{fb}$, $g_2 =
\mathit{ace}$ and $g_3 = d$. For the string $u = \mathit{ebcfa}$
we consider the guided rewrite sequence $\varrho = (\, (g_3,2)$,
$(g_1,0)$, $(g_2,1)$, $(g_1,0)$, $(g_1,3)$, $(g_1,3) \,)$.  The
associated reduction looks like
\begin{equation}
  \mathit{ebcfa} \; \toG_{g_3,2} \;
  \mathit{ebdfa} \; \toG_{g_1,0} \;
  \mathit{fbdfa} \; \toG_{g_2,1} \;
  \mathit{facea} \; \toG_{g_1,0} \;
  \mathit{fbcea} \; \toG_{g_1,3} \;
  \mathit{fbcfb} \; \toG_{g_1,3} \;
  \mathit{fbcfb}
  \label{eq-example-reduction}
\end{equation}
Recording what happens at all of the five positions of the string~$u$
yields, for this example, the slice sequence $\sigma = ( \slice_n
)^5_{n=1}$ given in the table at the left-hand side of
Figure~\ref{fig-five-slices}, where the slice sequence is visualized too.

\begin{figure}[htb]
  \centering
  \raisebox{2.75cm}{%
  \scalebox{0.85}{%
  \begin{math}
    \def\arraystretch{1.2}
    \begin{array}{|@{\ }c@{\ }|@{\ }c@{\ }|@{\quad}l@{\quad}|}
      \hline
      & I_n & (g_i,q_i)_{i \in I_n}  \\
      \hline
      \slice_1 & 2,4 & 
        2 \mapsto (g_1,1) ,\  
        4 \mapsto (g_1,1) \\
      \slice_2 & \ 2,3,4\ & 
        2 \mapsto (g_1,2),\ 
        3 \mapsto (g_2,1),\ 
        4 \mapsto (g_1,2) \\
      \slice_3 & 1,3  & 
        1 \mapsto (g_3,1) ,\ 
        3 \mapsto (g_2,2) \\
      \slice_4 & \ 3,5,6\ & 
        3 \mapsto (g_2,3) ,\
        5 \mapsto (g_1,1) ,\ 
        6 \mapsto (g_1,1) \\
      \slice_5 & 5,6 & 
        5 \mapsto (g_1,2) ,\ 
        6 \mapsto (g_1,2) \\
      \hline
    \end{array}
    \def\arraystretch{1.0}
  \end{math}
  } %% scalebox
  } %% raisebox
  \quad
  \qquad
  \includegraphics[scale=0.525]{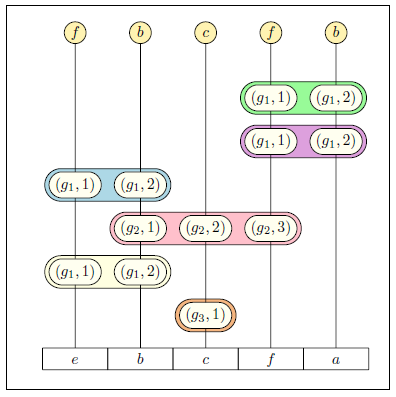}
  \caption{An example slice sequence}
  \label{fig-five-slices}
\end{figure}

For the choice of $I_1, \ldots, I_5$, the monotone partial injection
$\gamma_n$, $n = 1 \ldots 4$, maps every number to itself. It is
easily checked that all requirements of a slice sequence hold. The
ovals covering guide-offset pairs reflect the cuts as mappings between
to adjacent slices. However, they also comprise, in this situation
derived from a guided rewrite sequence, complete guides.  Note,
$\slice_1$ is a start slice, $\slice_5$~is an end slice. We have for
the slice sequence $\sigma = ( \slice_n )_{i=1}^5$ that $\yield(
\sigma ) = \yield( \slice_1 ) \cdot \cdots \cdot \yield( \slice_5) =
\mathit{fbcfb}$. Indeed, this coincides with the yield of the guided
rewrite sequence~$\varrho$ of~(\ref{eq-example-reduction}).
\end{example}

\noindent
The rest of this section is devoted to proving that the above
holds in general: Given a string and a set of guides, for every guided
rewrite sequence there exists a slice sequence and for every slice
sequence there exists a guided rewrite sequence. Moreover, the yield
of the guided rewrite sequence and slice sequence are the same.

\begin{theorem}
  \label{theorem-from-rewrites-to-slices}
  Let $\varrho = ( g_k , p_k )^r_{k = 1}$ be a guided rewrite
  sequence for a string~$u$. Then there exists a slice sequence
  $\sigma = ( \slice_n )^{\#u}_{n=1}$ for~$u$ such that $\yield(
  \sigma ) = \yield( \varrho )$.
\end{theorem}
\begin{proof}[Proof sketch]
  Induction on~$r$. If $\varrho$ is the empty rewrite sequence, we
  take for~$\sigma$ the slice sequence of $n$~empty slices. Suppose
  $\varrho$ is non-empty. Let $( u_k )_{k=0}^r$ be the sequence of
  strings induced by~$\varrho$. By induction hypothesis there exists a
  slice sequence~$\sigma'$ for the first $r{-}1$ steps
  of~$\varrho$. Suppose $u_{r{-}1} \togrpr u_r$. The slice
  sequence~$\sigma$ is obtained by extending the slices of~$\sigma'$
  from position $p_r{+}1$ to~$p_r{+}\#g_r$ with the pairs
  $(g_r,n{-}p_r)$. Then,
  \begin{displaymath}
    \def\arraystretch{1.2}
    \begin{array}{@{}lcl}
    \yield( \sigma ) & = &
    \yield( \sigma' [1,p_r] ) \cdot g_r[1,\#g_r] \cdot \yield(
    \sigma' [ p_r {+} \#g_r {+} 1, \#u] ) \\
    & = & 
    u_{r-1}[1,p_r] \cdot g_r \cdot u_{r-1}[ p_r {+} \#g_r
        {+} 1, \#u_{r-1} ]
    \: = \: u_r \: = \: \yield(\varrho)
    \end{array}
    \def\arraystretch{1.0}
  \end{displaymath}
  Verification of~$\sigma$ being a
  slice sequence for~$u$ requires transitivity of~$\R$.
\end{proof}

\noindent
In order to show the reverse of
Theorem~\ref{theorem-from-rewrites-to-slices} we proceed in a number of
stages. First we need to relate individual guide-offset pairs in
neighboring slices. For this purpose we introduce the
ordering~$\preccurlyeq$ on so-called chunks.

\begin{definition}
  \label{def-ordering}
  Let $\sigma = ( \slice_n )^{\#u}_{n=1}$ be a slice sequence
  for~$u$. Assume we have $\slice_n = ( \, g_{n,i} , q_{n,i} \, )_{i
    \in I_n}$, for $n = 1, \ldots, \#u$. Let $\gamma_n \colon I_n \to
  I_{n+1}$ be the cut for $\slice_n$ and~$\slice_{n+1}$, $1 \leqslant
  n < \#u$. Let $\X = \lc ( \, g_{n,i} ,q_{n,i} , i , n \, ) \mid 1
  \leqslant n \leqslant \#u, \ i \in I_n \rc$ be the set of chunks
  of~$\sigma$ and define the ordering~$\preccurlyeq$ on~$\X$ by
  putting $( \, g , q , i , n \, ) \preccurlyeq ( \, g' , {q
    \mkern1mu}', {i \mkern1mu}', n'
  \, )$ iff
  \begin{itemize}
  \item either $n' \geqslant n$ and there exist indexes $\ell_0, h_0,
    \ldots, \ell_{n'-n}, h_{n'-n}$ such that
    \begin{itemize}
    \item $\ell_k , h_k \in I_{n+k}$ and $\ell_k \leqslant h_k$, $0
      \leqslant k \leqslant n' - n$
    \item $h_k \in \dom( \gamma_{n+k} )$ and $\gamma_{n+k}( h_k ) =
      \ell_{k+1}$, $0 \leqslant k < n' - n$
    \item $\ell_0 = i$ and $h_{n'-n} = {i \mkern1mu}'$
    \end{itemize}
  \item or $n' \leqslant n$ and there exist indexes $\ell_0, h_0,
    \ldots, \ell_{n-n'}, h_{n-n'}$ such that
    \begin{itemize}
    \item $\ell_k , h_k \in I_{n'+k}$ and $\ell_k \leqslant h_k$, $0
      \leqslant k \leqslant n-n'$
    \item $\ell_k \in \dom( \gamma_{n'+k} )$ and $\gamma_{n'+k} (
      \ell_k ) = h_{k+1}$, $0 \leqslant k < n - n'$
    \item $h_0 = {i \mkern1mu}'$ and $\ell_{n - n'} = i$
    \end{itemize}
  \end{itemize}
\end{definition}

\blankline

\noindent
In the above setting with $n' \geqslant n$, we say that the sequence
$\ell_0, h_0$, $\ell_1, h_1$, $\ldots$, $\ell_{n'-n}, h_{n'-n}$ is
leading from $i \in I_n$ up to ${i \mkern1mu}' \in I_{n'}$. Likewise for the case
where $n' \leqslant n$.

For example, for the slice sequence $( \slice_i )_{i=1}^r$ of
Figure~\ref{fig-five-slices}, to identify the guide belonging to the
guide-offset pair~$(g_2,1)$ of slice~$\slice_2$, the pair is more
precisely represented by the chunk $(g_2, 1, 3, 2)$, for the pair is
associated with index~$3 \in I_2$ of slice~$\slice_2$. Since for the
cuts $\gamma_{\, 2} : I_2 \to I_3$ and $\gamma_{\, 3} : I_3 \to I_4$
we have $\gamma_{\, 2}(3)$ and $\gamma_{\, 3}(3) = 3$, we have
$(g_2,1,3,2) \preccurlyeq (g_2,2,3,3) \preccurlyeq (g_2,3,3,4)$ via
the sequence $3,3,3,3$ connects $(g_2,1)$ and~$(g_2,2)$, and $3,3,3,3$
connecting $(g_2,2)$ and~$(g_2,3)$. (Hence the combination of
sequences $3,3,3,3,3,3$ connects $(g_2,1)$ and~$(g_2,3)$ directly.)
As no jumps from a low index~$\ell$ to a high index~$h$ needs to be
taken, we also have $(g_2,1,3,2) \succcurlyeq (g_2,2,3,3) \succcurlyeq
(g_2,3,3,4)$. Thus $(g_2,1,3,2) \equiv (g_2,2,3,3) \equiv (g_2,3,3,4)
\rc$. In fact, $\lc (g_2,1,3,2) ,\, (g_2,2,3,3) ,\, (g_2,3,3,4) \rc$
is an equivalence class for~$\X$ corresponding to the guide~$g_2$
(cf.\ Lemma~\ref{lemma-partial-order}). Differently, we have
$(g_2,1,3,2) \preccurlyeq (g_1,2,6,5)$ relating $g_2$ to the fourth
occurrence of~$g_1$ via the sequence $3,3,3,3,3,5,5,5$, for
example. Since there is a jump here from $\ell_2 = 3$ to~$h_2 = 5$, we
do not have $(g_2,1,3,2) \succcurlyeq (g_1,2,6,5)$. This reflects that
apparently the rewrite with this occurrence of~$g_1$ is on top of part
of the rewrite using~$g_2$ as guide.

\blankline

\noindent
Given a slice sequence~$\sigma$, the ordering~$\preccurlyeq$ on the
chunks of~$\sigma$ in~$\X$ gives rise to a partial ordering on the set
$\X / \mathord{\equiv}$ of equivalence classes of chunks. As we will
argue, the equivalence classes correspond to guides and their ordering
corresponds to the relative order in which the guides occur in a
rewrite sequence~$\varrho$ having the same yield as the slice
sequence~$\sigma$.

\begin{lemma}
  \label{lemma-partial-order}
  \begin{itemize}
  \item [(a)] The relation~$\preccurlyeq$ on~$\X$ is reflexive and
    transitive.
  \item [(b)] The relation $\equiv$ on~$\X$ such that $x \equiv y \iff
    x \preccurlyeq y \land y \preccurlyeq x$ is an equivalence
    relation.
  \item [(c)] The ordering~$\preccurlyeq$ on~$\X / \mathord{\equiv}$
    induced by~$\preccurlyeq$ on~$X$ by
  \begin{math}
    [x] \preccurlyeq [y] \iff \exists x' \in [x] \, \exists y' \in [y]
    \colon x' \preccurlyeq y'
  \end{math},
  makes $\X / \mathord{\equiv}$ a partial order.  
  \qed
  \end{itemize}
\end{lemma}

\sloppypar{%
\noindent
The next lemma describes the form of the equivalence class holding a
chunk~$x = (g,q,i,n)$. Using the cuts, equivalent chunks can be found
backwards up to position~$n{-}q{+}1$ and forward up to position
$n{-}q{+}\#g$. These chunks together, $(g, 1, i_{n{-}q{+}1}, n{-}q{+}1)$,
$\ldots$, $(g, q, i_n, n)$, $\ldots$, $(g, \#g, i_{n{-}q{+}\#g},
n{-}q+\#g)$ span the guide~$g$ that is to be applied, in the rewrite
sequence to be constructed.
} %% end sloppypar

\begin{lemma}
  \label{lemma-form-of-equivalence-classes}
  Let $\sigma = ( \slice_n )^{\#u}_{n=1}$ be a slice sequence for a
  string~$u$. Let $\X = \lc ( \, g_{n,i} ,q_{n,i} , i , n \, ) \mid 1
  \leqslant n \leqslant \#u, \ i \in I_n \rc$ be the set
  of chunks and choose $x \in \X$, say $x = (g,q,i,n)$. Put $p =
  n{-}q$. Then there exist $j_1 \in I_{p{+}1}$, $\ldots\,$, $j_{\#g}
  \in I_{p{+}\#g}$ such that
  \begin{math}
    [x] = \lc ( g , s , j_s , p+s ) \mid 1 \leqslant s \leqslant \#g \rc
  \end{math}.
  \qed
\end{lemma}

\noindent
We are now in a position to prove the reverse of
Theorem~\ref{theorem-from-rewrites-to-slices}.

\begin{theorem}
  \label{theorem-from-slices-to-rewrites}
  Let $\sigma$ be a slice sequence for a string~$u$. Then there exists
  a guided rewrite sequence~$\varrho$ for~$u$ such that $\yield(
  \varrho ) = \yield( \sigma )$.
\end{theorem}
\begin{proof}
  Suppose $\sigma = ( \, \slice_n \, )^{\#u}_{n=1}$, $\slice_n = ( \,
  g_{i,n} , q_{i,n} \, )_{i \in I_n}$, for $n = 1 , \ldots, \#u$, and
  let $\X = \lc ( \, g_{n,i} ,q_{n,i} , i , n \, ) \mid 1 \leqslant n
  \leqslant \#u, \ i \in I_n \rc$ be the corresponding set of
  chunks. We proceed by induction on $\size{\X}$. Basis, $\size{\X} =
  0$: In this case every slice is empty and $\yield(\sigma) = \yield(
  \, \slice_1 \, ) \cdots \yield( \, \slice_{\#u} \, ) = u[1] \cdots
  \cdot u[\#u] = u$ and the empty guided rewrite sequence for~$u$ has
  also yield~$u$.

  Induction step, $\size{\X} \mathop{>} \mkern1mu 0$: Clearly, $\X /
  \mathord{\equiv}$ is finite and therefore we can choose, by
  Lemma~\ref{lemma-partial-order}, $x \in \X$ such that $[x]$~is
  maximal in $\X / \mathord{\equiv}$ with respect
  to~$\preccurlyeq$. By Lemma~\ref{lemma-form-of-equivalence-classes}
  we can assume $[x] = \lc {( g , s, i_s , p+s)} \mid 1 \leqslant s
  \leqslant \#g \rc$ for suitable~$p$ and indexes $i_s \in I_{p{+}s}$,
  for $s = 1, \ldots, \#g$. Note, by maximality of~$[x]$, the
  indexes~$i_s$ must be the maximum of~$I_{p{+}s}$. In particular,
  $\yield( \, \sigma \, ) [p+s] = \yield( \, \slice_{p{+}s} \, ) =
  g[s]$, for $s = 1, \ldots, \#g$.

  Now, consider the slice sequence $\sigma' = ( \, \slice'_n \,
  )^{\#u}_{n=1}$ where
  \begin{displaymath}
    \slice'_n = \left \lbrace
    \begin{array}{@{\ }lcl}
      \slice_n && 
      \text{ for $n = 1 , \ldots, p$ and $n = p {+} \#g {+} 1 ,
        \ldots, \#u$} \\
      ( \, g_{i,n} , q_{i,n} \, )_{i \in I_n \backslash \lbrace i_{n{-}p}
        \rbrace } &&
      \text{ for $n =  p {+} 1 , \ldots, p {+} \#g$}
    \end{array}
    \right .
  \end{displaymath}
  So, the slice sequence~$\sigma'$ is obtained from the slice
  sequence~$\sigma$ by leaving out the guide-offset pairs related to
  the particular occurrence of~$g$.

  Let~$\X'$ be the set of chunks of~$\sigma'$. Then $\size{\X'}
  \mathop{<} \mkern1mu \size{\X}$. By induction hypothesis we can find
  a guided rewrite sequence $\varrho' = ( \, g'_k , p'_k \, )^r_{k=1}$
  for~$u$ such that $\yield( \, \varrho' \, ) = \yield( \, \sigma' \,
  )$. Define the guided rewrite sequence $\varrho = ( \, g_k , p_k \,
  )^{r+1}_{k=1}$ by $g_k = g'_k$, $p_k = p'_k$ for $k = 1, \ldots, r$
  and $g_{r{+}1} = g$, $p_{r{+}1} = p$.  We have $0 \leqslant p
  \leqslant \#u {-} \#g$ and $u[ p{+}1 , p{+}\#g ] \mbR g$ since
  $\slice_{p{+}1} , \ldots , \slice_{p{+}\#g}$ are slices for
  $u[p{+}1], \ldots, u[p{+}\#g]$, respectively. So, $\varrho$~is a
  well-defined guided rewrite sequence for~$u$.

  It holds that $\yield( \, \varrho' \, ) \togp \yield( \, \varrho \,
  )$ as~$\varrho$ extends~$\varrho'$ with the pair~$(g,p)$. Therefore,
  \begin{displaymath}
    \yield( \varrho )[n] = \left \lbrace
    \begin{array}{lcl}
      \yield( \varrho' )[n] && \text{for $n = 1, \ldots, p$ and $n = 
        p {+} \#g {+} 1, \ldots, p {+} \#g$} \\
      g[n{-}p] && \text{for $n = p {+} 1, \ldots, p {+} \#g$}
    \end{array}
    \right .
  \end{displaymath}
  From this it follows, for any index~$n$, $1 \leqslant n \leqslant p$
  or $p{+}\#g{+}1 \leqslant n \leqslant \#u$, that
  \begin{math}
    \yield( \, \varrho \, )[n] = \yield( \, \varrho' \, )[n] = \yield(
    \, \sigma' \, )[n] = \yield( \, \sigma \, )[n],
  \end{math}
  and for any index~$n$, $p {+} 1 \leqslant n \leqslant p {+} \#g$, that
  \begin{math}
    \yield( \, \varrho \, )[n] = g[n{-}p] = \yield( \, \sigma \, )[n]
  \end{math}.
  As $\size{\yield( \, \varrho \, )} = \size{\yield(\, \sigma \, )} =
  \#u$, we obtain $\yield( \, \varrho \, ) = \yield( \, \sigma \, )$,
  as was to be shown.  
\end{proof}

\blankline

\noindent
For the slice sequence~$( \slice_i )_{i=1}^5$ of
Figure~\ref{fig-five-slices} we have the following equivalence classes
of chunks:
\begin{displaymath}
  \def\arraystretch{1.2}
  \begin{array}{@{}lcl@{\qquad}lcl}
    G_{\, 3} & = & \lc (g_3,1,1,3) \rc &
    G_{\, 2} & = & \lc (g_2,1,3,2) ,\, (g_2,2,3,3) ,\, (g_2,3,3,4) \rc \\
    G_{\, 1}^{\, 1} & = & \lc (g_1,1,2,1) ,\, (g_1,2,2,2) \rc &
    G_{\, 1}^{\, 3} & = & \lc (g_1,1,5,4) ,\, (g_1,2,5,5) \rc \\
    G_{\, 1}^{\, 2} & = & \lc (g_1,1,4,1) ,\, (g_1,2,4,2) \rc &
    G_{\, 1}^{\, 4} & = & \lc (g_1,1,6,4) ,\, (g_1,2,6,5) \rc
  \end{array}
  \def\arraystretch{1.3}
\end{displaymath}
Moreover, $G_{\, 3} \preccurlyeq G_{\, 1}^{\, 1} \preccurlyeq G_{\,
  2}$, $G_{\, 2} \preccurlyeq G_{\, 1}^{\, 2}$ and $G_{\, 2}
\preccurlyeq G_{\, 1}^{\, 3} \preccurlyeq G_{\, 1}^{\, 4}$. A possible
linearization is $G_{\, 3} \preccurlyeq G_{\, 1}^{\, 1} \preccurlyeq G_{\,
  2} \preccurlyeq G_{\, 1}^{\, 3} \preccurlyeq G_{\, 1}^{\, 4}
\preccurlyeq G_{\, 1}^{\, 2}$. This corresponds to the rewrite
sequence
\begin{displaymath}
  \mathit{ebcfa} \toG_{g_3,2} 
  \mathit{ebdfa} \toG_{g_1,0} 
  \mathit{fbdfa} \toG_{g_2,1} 
  \mathit{facea} \toG_{g_1,3} 
  \mathit{facfb} \toG_{g_1,3} 
  \mathit{facfb} \toG_{g_1,0} 
  \mathit{fbcfb}
\end{displaymath}
Note that the yield $\mathit{fbcfb}$ of this rewrite sequence is the
same as the yield of the sequence~(\ref{eq-example-reduction}) of
Example~\ref{ex-slice-sequence}. However, here the second rewrite
with~$g_1$ of~(\ref{eq-example-reduction}) has been moved to the
end. This does not effect the end result as the particular rewrites do
not overlap.

%% \noindent
%% In the next section we will have occasion to restrict to a finite set
%% of relevant slices. The next lemma prepares the grounds for
%% this. Given the apparatus available its proof is relatively
%% straightforward. However, via the correspondence of rewrite sequences
%% and slice sequences, this translates to the observation that every
%% sequence of guides rewrites can be trimmed to a sequence of guided
%% rewrites without repetition having the same result.

%% \begin{lemma}
%%   \label{lemma-repetition-free}
%%   \begin{alphalist}
%%   \item If $\sigma$ is a slice sequence for a string~$u$, then there
%%     exists a slice sequence~$\sigma'$ for~$u$ of repetition free
%%     slices only, such that $\yield(\sigma) = \yield(\sigma')$.
%%   \item If $\varrho$ is a rewrite sequence for a string~$u$, then
%%     there exists a repetition-free rewrite sequence~$\varrho'$ such
%%     that $\yield(\varrho) = \yield(\varrho')$.
%%   \end{alphalist}
%% \end{lemma}

%% \begin{proof}[Sketch]
%%   EV: to be written, but will be exploiting the two theorems
%%   \ref{theorem-from-rewrites-to-slices}
%%   and~\ref{theorem-from-slices-to-rewrites}.
%%   %
%%   \qed
%% \end{proof}

\section{Guided rewriting preserves regularity}
\label{sec-preservation}

Given a language~$L$ and a set of guides~$G$, the language~$L_G$ is
given as the set $\lc v \in \Sigmaast \mid \exists u \in L \colon u
\toGast v \rc$. One of the main results of this paper,
Theorem~\ref{thmmain} formulated in
Section~\ref{sec-guided-rewriting}, states that if $L$~is regular
than~$L_G$ is regular too. We will prove the theorem by constructing a
non-deterministic finite automaton accepting~$L_G$ from a
deterministic finite automaton accepting~$L$. The proof exploits the
correspondence of rewrite sequences and slice sequences,
Theorem~\ref{theorem-from-rewrites-to-slices} and
Theorem~\ref{theorem-from-slices-to-rewrites}. First we need an
auxiliary result to assure finiteness of the automaton for~$L_G$.

\begin{lemma}
  \label{lemma-finitely-many-slices}
  Let $G$ be a finite set of guides. Let $Z = \lc \slice \mid
  \text{$\slice$ repetition-free slice for~$a$ and~$G$, $a \in
    \Sigma$} \rc$. Then $Z$~is finite. Moreover, for every string~$u$
  and every rewrite sequence~$\varrho$ for~$u$, there exists a slice
  sequence~$\sigma$ for~$u$ consisting of slices from~$Z$ only such
  that $\yield(\sigma) = \yield(\varrho)$.
\end{lemma}

\begin{proof}[Proof sketch]
  Finiteness of~$Z$ is immediate: there are finitely many guide-offset
  pairs $(g,q)$, hence finitely many repetition-free finite sequences
  of them. Thus, there are only finitely many repetition-free slices.

  Now, let $\varrho$ be a rewrite sequence for a string~$u$. By
  Theorem~\ref{theorem-from-rewrites-to-slices} we can choose a slice
  sequence~$\sigma'$ such that $\yield(\sigma') = \yield(\varrho)$.
  Suppose $\sigma' = ( \slice_n )_{n=1}^{\#u}$ and $\slice_n = (
  g_{i,n}, q_{i,n} )_{i \in I_n}$ for $n = 1, \ldots, \#u$. By
  Lemma~\ref{lemma-form-of-equivalence-classes} it follows that given
  a repeated guide-offset pair $(g,q)$, say $(g,q) =
  (g_{i,n},q_{i,n})$ and $(g,q) = (g_{j,n},q_{j,n})$ for indexes $i <
  j$ in~$I_n$, we can delete the complete equivalence class of
  $(g_i,q_i,i,n)$ from slices $\slice_{n-{q}{+}1}$
  to~$\slice_{n{-}q{+}\#g}$, while retaining a slice
  sequence. 
  In fact, we are removing the `lower' occurrence of the guide~$g$.
  Moreover, the resulting slice sequence has the same yield
  as for all slices the topmost guide-offset pair remains
  untouched. The existence of a repetition-free slice
  sequence~$\sigma$ such that $\yield(\sigma) = \yield(\sigma')$,
  hence $\yield(\sigma) = \yield(\varrho)$, then follows by induction
  on the number of repetitions.
\end{proof}

\noindent
As a corollary we obtain that every rewrite sequence has a
repetition-free equivalent, an intuitive result which requires some
technicalities though to obtain directly.

We are now prepared to prove that guided rewriting preserves
regularity.

\begin{proof}[Proof of Theorem~\ref{thmmain}]
  Without loss of generality $\varepsilon \notin L$.
  Let $M = ( \Sigma, Q, {\to}, q_o, F )$ be a DFA accepting~$L$. We
  define the NFA $M' = ( \Sigma, Q', {\toprime}, q_0, F' )$ as
  follows: Let $q_F$ be a fresh state. Put $Q' = Q \cup ( Q
  \mathop{\times} Z ) \cup \lbrace q_F \rbrace$ with $Z$ as given by
  Lemma~\ref{lemma-finitely-many-slices}, $F' = \lbrace q_F 
  \rbrace$ and
  \begin{displaymath}
    \begin{array}{r@{\,}c@{\,}l@{\quad}l}
      q_0 & \toprimeepsilon & q_0 \times \zeta & \text{if $\zeta$ is a
        start slice} \\

      q \times \zeta & \toprimeb & \qprime \times \zeta' & \text{if $q \toa
        \qprime$, $a \mbR \zeta$, $\yield(\zeta) = b$, $\zeta \leadsto
        \zeta'$} \\ 

      q \times \zeta & \toprimeb & q_F & \text{if $\exists \mkern2mu
        \qprime \colon q \toa \qprime\in F$, $a \mbR \zeta$,
        $\yield(\zeta) = b$, $\zeta$ is an end slice}
    \end{array}
  \end{displaymath}
  Note, by Lemma~\ref{lemma-finitely-many-slices}, $Q'$ is a finite
  set of states. The automaton~$M'$ has only one final state,
  viz.~$q_F$.

  Suppose $v \in L_G$. Then there exist $u = a_1 \cdots a_s \in L$, a
  rewrite sequence $\varrho = ( g_k, p_k )_{k=1}^r$ and strings $u_0,
  u_1, \ldots, u_r$ such that $u = u_0$, $u_{k{-}1} \togkpk u_k$ for
  $k = 1, \ldots, r$, and $v = u_r$. Let, by
  Theorem~\ref{theorem-from-rewrites-to-slices} and
  Lemma~\ref{lemma-finitely-many-slices}, $\sigma$~be a slice sequence
  for~$u$ of repetition-free slices with $\yield(\sigma) =
  \yield(\varrho)$. Say $\sigma = ( \slice_n )_{n=1}^{\#u}$ and $\slice_n = (
  g_{i,n}, q_{i,n} )_{i \in I_n}$ for $n = 1, \ldots, \#u$. Let $q_0
  \toaone q_1 \cdots \toas q_{s} \in F$ be an accepting computation
  of~$M$ for~$u$. Then $q_0 \toprimeepsilon q_0 \mathord{\times}
  \slice_1 \toprimebone \cdots q_{s{-}1} {\times} \slice_{s}
  \toprimebs q_F$ is an accepting computation of~$M'$. Since we have
  $b_1 \cdots b_s$ = $\yield(\slice_1)$ $\cdots$ $\yield(\slice_s)$ =
  $\yield(\sigma) = v$, it follows that $v \in \L(M')$. So, $L_G
  \subseteq \L(M')$.

  Let $v = b_1 \cdots b_s$ be a string in~$\L(M')$. Given the
  definition of the transition relation on~$M'$, we can find states
  $q_0, q_1, \ldots, q_{s{-}1}$, repetition-free slices $\slice_1,
  \ldots \slice_{s}$ such that $\slice_n \leadsto \slice_{n{+}1}$ for
  $n = 1, \ldots, s{-}1$, and a computation $q_0 \toprimeepsilon q_0
  \mathord{\times} \slice_1 \toprimebone \cdots q_{s{-}1}
  \mathord{\times} \slice_s \toprimebs q_F$. Thus, there exist a final
  state~$q_s$ and a computation $q_0 \toaone q_1 \cdots q_{s{-}1}
  \toas q_s \in F$ such that $a_n \mbR \slice_s$ for $n = 1, \ldots,
  s$, i.e.\ $\slice_n$ is a slice for~$a_n$. Put $u = a_1 \cdots
  a_s$. Then $u \in L$, $( \slice_n )_{n=1}^{\#u}$ is a slice sequence
  for~$u$ and $\yield(\sigma) = v$. By
  Theorem~\ref{theorem-from-slices-to-rewrites} we can find a rewrite
  sequence~$\varrho$ for~$u$ such that $\yield(\varrho) =
  \yield(\sigma) = v$. It follows that $u \toGast v$ and $v \in
  L_G$. Thus, $\L(M') \subseteq L_G$. We conclude that $L_G = \L(M')$
  and regularity 
  of~$L_G$ follows.
\end{proof}

\noindent
Since $L \subseteq L_G$ the automaton~$M'$ should accept any word $a_1
\ldots a_s \in L$, $s > 0$. This can be verified as
follows. Let~$\zeta_{\, i}$ be the empty slice for~$a_i$, $i = 1
\ldots s$. Then $a_i \mbR \zeta_{\, i}$, i.e.\ $a_i = \yield(
\zeta_{\, i} )$, which holds by definition. Moreover, $\zeta_{\,
  1}$~is a start slice, $\zeta_{\, i} \leadsto \zeta_{\, i{+}1}$ for
$i = 1 \ldots {s{-}1}$, and $\zeta_{\, s}$~is an end slice. It follows
that we can turn an accepting computation of~$M$, say $q_0
\xrightarrow{a_1} q_1 \xrightarrow{a_2} {} \cdots {} \xrightarrow{a_s}
q_s \in F$ into an accepting computation of~$M'$:
\begin{math}
  q_0 \xrightarrow{\varepsilon}\mkern-3mu{}' q_0 \times \zeta_{\, 1}
  \xrightarrow{a_1}\mkern-3mu{}' q_1 \times \zeta_{\, 2}
  \xrightarrow{a_2}\mkern-3mu{}' {} \cdots {}
  \xrightarrow{a_{s{-}1}}\mkern-3mu{}' q_{s{-}1} \times \zeta_{\, s}
  \xrightarrow{a_s}\mkern-3mu{}' q_F \, .
\end{math}

%% \noindent
%% Regarding the proof we note that the construction also works if
%% started from an NFA for~$L$ if precautions to deal with
%% $\varepsilon$-transitions are being made. The transition relation
%% of~$M'$ then has the following clauses:
%% \begin{displaymath}
%%   \begin{array}{r@{\,}c@{\,}l@{\quad}l}
%%     q_0 & \toprimeepsilon & \qprime \times \zeta & \text{if $q_0
%%       \toastepsilon q$, $\zeta$ is a start slice} \\
    
%%     q \times \zeta & \toprimeb & {q \mkern1mu}'' \times \zeta' & \text{if
%%       $\exists \mkern2mu \qprime \colon q \toastepsilon \qprime \toa
%%       {q \mkern1mu}''$, $a
%%       \mbR \zeta$, $\yield(\zeta) = b$, $\zeta \leadsto \zeta'$} \\

%%     q \times \zeta & \toprimeb & q_F & \text{if $\exists \mkern2mu
%%       \qprime, {q \mkern1mu}'', {q \mkern1mu}''' \colon q
%%       \toastepsilon \qprime \toa {q \mkern1mu}'' \toastepsilon
%%       {q \mkern1mu}''' \in F$, $a \mbR \zeta$}, \\
    
%%     &&& \phantom{\text{if }} \text{$\yield(\zeta) = b$, $\zeta$ is an
%%       end slice}
%%   \end{array}
%% \end{displaymath}
%% Therefore, if we would consider a push-down automaton the above
%% $\varepsilon$-closure would also affect the changes of the
%% stack. However, as this involves superposition of finitely many
%% transitions only, such can be done relatively straightforward. It
%% therefore comes natural to conjecture, but we have no proof at
%% present, that the theorem also extends to context-free languages,
%% i.e.\ if $L$ is a context-free language and $G$~is a finite set of
%% guides then $L_G$ is a context-free language too.

\blankline

\noindent
We now return to a proof of Theorem~\ref{thmid} formulated in
Section~\ref{sec-ins-del} for which we want to apply
Theorem~\ref{thmmain}. For the latter theorem to apply we need a
preparatory transformation. The point is, in the setting of guided
insertion/deletion, strings are allowed to grow or shrink while guided
insertions and deletions are being applied, whereas in the setting of
guided rewriting the strings do not change length.

The key idea of the transformation is that every group of~$0$'s is
compressed to a single symbol. Let a language~$L$ over~$\Sigma$ and a
number~$k$ be given by Theorem~\ref{thmid}. So, $L$ does not contain
strings with $k$~or more~$0$s. We introduce $k$~fresh symbols $0_0,
0_1, \ldots, 0_{k-1}$. Put $\Theta = \lc 0_0 ,\, 0_1 ,\, \ldots ,\,
0_{k-1} \rc$. For any string~$u$ over $\Sigma$ not containing the
substring~$0^{\mkern1mu k}$, i.e.\ not containing $k$~or more zeros,
we define the string $\bar{u}$ over the alphabet $\overline{\Sigma} =
( \Sigma \setminus \singleton{0} ) \cup \{0_0, 0_1, \ldots, 0_{k-1}\}$
that is obtained from~$u$ by replacing every maximal pattern
$0^{\mkern1mu i}$ by the single symbol~$0_i$.  Note, between two
consecutive non-zero letters~$ab$ the symbol~$0_0$ is interspersed.
For instance, for $k \geqslant 3$, $\overline{10023} = 1 0_2 2 0_0
3$. Also note, that the compression scheme constitutes a 1--1
correspondence of~$\Sigmaast \cap \lc w \mid \text{$w$ has no
  substring~$0^{\mkern1mu k}$} \rc$ and $\bigl( \, \Theta \cdot
\Sigmanod \, \bigr)^\ast \cdot \Theta$. 
%% Here, the bracket notation $[ \, \Theta \, ]$ indicates that a
%% trailing element from~$\Theta$ is optional. More specifically we
%% have $0^{\mkern1mu i_0} a_0 \mkern1mu 0^{\mkern1mu i_1} \ldots
%% 0^{\mkern1mu i_n} a_n \cdot [ \, 0^{i_{n+1}} \, ]$ corresponds to
%% $0_{\mkern1mu i_0} a_0 \mkern1mu 0_{\mkern1mu i_1} \ldots
%% 0_{\mkern1mu i_n} a_n \cdot [ \, 0_{i_{n+1}} \, ]$.

Next, we show that the above operation of compressing groups
of~$0$s preserves regularity using basic closure properties of the
class of regular languages, cf.~\cite[Section~3]{HU79}.

\blankline

\begin{lemma}
  \label{lemid}
  Let $L$ be a language without strings containing $0^{\mkern1mu k}$
  and let $\overline{L} = \{ \bar{u} \mid u \in L\}$. Then
  $\overline{L}$ is regular if and only if $L$ is regular.
\end{lemma}
%% \begin{proof}
%%   A finite automaton for~$L$ over the alphabet~$\Sigma$ can be
%%   transformed into a finite automaton over~$\overline{\Sigma}$
%%   accepting~$\overline{L}$ by adding $0_i$-transitions for every pair
%%   of states~$s,t$ for which $t$ can be reached from~$s$ by~$i$
%%   $0$-steps, for $i > 0$ and removing all $0$-transitions, to deal
%%   with strings~$0^i$, $i > 0$, and by
%%   adding two fresh states, $s'$~and~$t'$~say, for every transition~$s
%%   \xrightarrow{a} t$ together with transitions $s \xrightarrow{a} s'$,
%%   $s' \xrightarrow{0_0} t'$, $t' \xrightarrow{b} u$ if $t
%%   \xrightarrow{b} u$, for each~$b \in \Sigmanod$, and removing these
%%   transitions for~$t$, to deal with the correspondence of~$ab$ and $a
%%   \mkern1mu 0_0b$.

%%   Reversely, a finite automaton for~$\overline{L}$ with
%%   alphabet~$\overline{\Sigma}$ can be transformed into a finite
%%   automaton over~$\Sigma$ by replacing each $0_i$-transition, for $i >
%%   0$, by a sequence of $i$~transitions labeled with~$0$, adding fresh
%%   states if necessary. Transitions $s \xrightarrow{0_0} t$ are
%%   replaced by an $\varepsilon$-transition.
%% \end{proof}
\begin{proof}
  The language~$L$ is the homomorphic image of~$\overline{L}$ for $h
  \colon \overline{\Sigma}^\ast \to \Sigma^\ast$ with $h(0_i) = 0^{\,
    i}$ and $h(a) = a$ otherwise. So, if $\overline{L}$ is regular, so
  is~$L$. Reversely, $\overline{L} = { ( \Theta \cdot \Sigma )^\ast
    \cdot \Theta } \cap h^{-1}(L)$. Hence, if $L$ is regular, so
  is~$\overline{L}$.
\end{proof}

\noindent
With the above lemma in place we can give a proof of the preservation
of regularity by guided insertion/deletion.

\begin{proof}[Proof of Theorem~\ref{thmid}]
  Let $k$ be as given by the statement of the theorem. Obtain
  $\overline{L}$ by applying the compression of strings~$0^{\mkern1mu
    i}$, for~$i < k$, changing from the alphabet~$\Sigma$
  to~$\overline\Sigma$, as introduced above. By Lemma~\ref{lemid} we
  then have that $\overline{L}$ is regular. Let~$\overline{G}$ be
  obtained from~$G$, again by compression of strings~$0^{\mkern1mu
    i}$, for~$i < k$. Then $\overline{G}$ is a finite set of
  guides with respect to~$\overline\Sigma$.
  Now let the adjustment relation~$\R$ be the equivalence relation
  on~$\overline\Sigma$ generated by $0_{\mkern1mu i} \R 0_{\mkern1mu
    j}$, $0 \leqslant i ,\, j < k$. By Theorem~\ref{thmmain} we obtain
  that $\overline{L}_{\overline{G}}$ is regular.

  Next we note that if $u \idr v$ with respect to~$\Sigma$, then
  $\bar{u} \toGbar \bar{v}$ with respect to~$\overline\Sigma$. Vice
  versa, if $\bar{u} \toGbar \bar{v}$ and there exist (unique) $u$
  and~$v$ such that $u, v$ map to $\bar{u}, \bar{v}$ under
  compression, then $u \idr v$. It follows that
  $\overline{L}_{\overline{G}}$ and $\overline{L_{i/d}}$
  coincide. Finally, by another application of Lemma~\ref{lemid}, we
  conclude that $L_{i/d}$ is regular.
\end{proof}

\section{Related work and concluding remarks}
\label{seccr}

In this paper we have discussed abstract concepts of guided rewriting:
a more flexible notion focusing on insertions and deletions of a dummy
symbol, another more strict notion based on an equivalence
relation. Given a language~$L$ we considered the extended languages
$L_{i/d}$ and~$L_{\mkern1mu G}$ comprising the closure of~$L$ for the
two types of guided rewriting with guides from a finite set~$G$. In
particular, as our main result we proved that these closures preserve
regularity. For doing so we investigated the local effect of guided
rewriting on two consecutive string positions, leading to a novel
notion of a slice sequence. Finally, the theorem for adjustment-based
rewriting was proved by an automaton construction exploiting a slice
sequence characterization of guided rewriting. Via a compression
scheme for strings of dummy symbols, the theorem for guided
insertion/deletion followed.

Preservation of regularity by closing a language with respect to a
given notion of rewriting arises as a natural question. In Section
\ref{sec-ins-del} we observed that by closing the regular language
$\L( \, (ab)^\ast \,)$ under rewriting with respect to the single
rewrite rule $ba \to ab$ the resulting language is not regular. So, by
arbitrary string rewriting regularity is not necessarily preserved. A
couple of specific rewrite formats have been proposed in the
literature. In \cite{HW04} it was proved that regularity is preserved
by deleting string rewriting, where a string rewriting system is
called deleting if there exists a partial ordering on its alphabet
such that each letter in the right-hand side of a rule is less than
some letter in the corresponding left-hand side. In \cite{L08} it was
proved that regularity is preserved by so-called period expanding or
period reducing string rewriting.  When translated to the setting
of~\cite{Zan10:lata}, as also touched upon in
Section~\ref{sec-ins-del}, our present notion of guided insertions and
deletions allows for simultaneous insertion and deletion of the dummy
symbol. A phenomenon also supported by biological
findings. Remarkably, the more liberal guided insertion/deletion
approach preserves regularity, whereas in the more restricted
mechanism of~\cite{Zan10:lata}, not mixing insertions and deletions
per rewrite step, regularity is not preserved.  As another striking
difference with the mechanism of~\cite{Zan10:lata}, for that format it
was shown that strings~$u,v$ of length~$n$ exist satisfying $u
\Rightarrow^* v$, but the length of the reduction is at least
exponential in~$n$.  In our present format this is not the case: we
expect that our slice characterization of guided rewriting serves to
prove, that if $u \Rightarrow^* v$ then there is always a
corresponding reduction of length linear in the length of $u$
and~$v$. Details have not been worked out yet.

The notion of splicing, inspired by DNA recombination, has been
proposed by Head in~\cite{Hea87:bmb}. A so-called splicing rule is a
tuple $r = (u_1,v_1;u_2,v_2)$. Given two words $w_1 = x_1 u_1 v_1 y_1$
and $w_2 = x_2 u_2 v_2 y_2$ the rule~$r$ produces the word $w = x_1
u_1 v_2 y_2$. So, the word~$w_1$ is split in between $u_1$
and~$v_1$, the word~$w_2$ in between $u_2$ and~$v_2$ and the resulting
subwords $x_1 u_1$ and~$v_2 y_2$ are recombined into the
word~$w$. 
%% Given a set of splicing rules~$S$, we may write $w_1 , w_2
%% \Rightarrow_S w$ if $w_1$ and~$w_2$ can be recombined into~$w$ using a
%% splicing rule from~$S$.
For splicing a closure result, reminiscent to
the one for guided rewriting considered in this paper, has been
established. Casted in our terminology, if $L$ is a regular language
and $S$~is a finite set of splicing rules, then $L_S$~is regular too,
cf.~\cite{CH91:dam,Pix96:dam}. Here, $L_S$ is the least language
containing~$L$ and closed under the splicing rules of~$S$.
%% , i.e.\ $L_S$
%% is minimal such that $L \subseteq L_S$ and ${w_1,w_2 \in L_S} \land
%% {w_1,w_2 \Rightarrow_S w}\Longrightarrow {w \in L_S}$.

The computational power of a variant of insertion-deletion systems was
studied in~\cite{TY03}. There deletion means that a string $u \alpha
v$ is replaced by $u v$ for a predefined finite set of triples $u,
\alpha, v$, while by insertion a string $u v$ is replaced by $u \alpha
v$ for another predefined finite set of triples $u, \alpha, v$. This
notion of insertion-deletion is quite different from ours, and seems
less related to biological RNA editing.  In the same vein are the
guided insertion/deletion systems of~\cite{BBD07}. There a hierarchy
of classes of insertion/deletion systems and related closure
properties are studied. Additionally, a non-mixing insertion/deletion
system that models part of the RNA-editing for kinetoplastids is
given.  A rather different application of term rewriting in the
setting of RNA is reported in~\cite{FWHH05}, where the rewrite engine
of Maude is exploited to predict the occurrence of specific patterns
in the spatial formation of RNA, with competitive precision compared
to techniques that are more frequently used in bioinformatics.

Possible future work includes investigation of preservation of
context-freedom and of lifting the bound on the number of
consecutive~$0$'s in Theorem \ref{thmid}. More specifically, for a
context-free language~$L$, does it hold, for a finite set of
guides~$G$, that $L_G$ is context-free too? Considering the set of
guides, a generalization to regular sets~$G$ is worthwhile
studying. Note that the counter-example given in
Section~\ref{sec-guided-rewriting} involves a non-regular set of
guides. So, if $L$~is regular and $G$~is regular, do we have that
$L_G$ is regular? Similarly for $L$~context-free. We also plan to
consider guided rewriting based on other types of adjustment
relations. In particular, rather than comparing strings
symbol-by-symbol, one can consider two strings compatible if they map
to the same string for a chosen string homomorphism. A prime example
would be the erasing of the dummy~$0$ in the context of
Section~\ref{sec-ins-del} for which we conjecture a variant of
Theorem~\ref{thmmain} to hold.

\bibliographystyle{eptcs}
\bibliography{mecbic2012}

\begin{thebibliography}{10}
\providecommand{\bibitemdeclare}[2]{}
\providecommand{\surnamestart}{}
\providecommand{\surnameend}{}
\providecommand{\urlprefix}{Available at }
\providecommand{\url}[1]{\texttt{#1}}
\providecommand{\href}[2]{\texttt{#2}}
\providecommand{\urlalt}[2]{\href{#1}{#2}}
\providecommand{\doi}[1]{doi:\urlalt{http://dx.doi.org/#1}{#1}}
\providecommand{\bibinfo}[2]{#2}

\bibitemdeclare{article}{ATS97}
\bibitem{ATS97}
\bibinfo{author}{J.D. \surnamestart Alfonzo\surnameend},
  \bibinfo{author}{O.~\surnamestart Thiemann\surnameend} \&
  \bibinfo{author}{L.~\surnamestart Simpson\surnameend} (\bibinfo{year}{1997}):
  \emph{\bibinfo{title}{The Mechanism of Insertion/Deletion {RNA} Editing in
  Kinetoplastid Mitochondria}}.
\newblock {\sl \bibinfo{journal}{Nucleic Acids Research}}
  \bibinfo{volume}{25}(\bibinfo{number}{19}), pp. \bibinfo{pages}{3751--3759},
  \doi{10.1093/nar/25.19.3571}.

\bibitemdeclare{article}{BBD07}
\bibitem{BBD07}
\bibinfo{author}{F.~\surnamestart Biegler\surnameend}, \bibinfo{author}{M.J.
  \surnamestart Burrell\surnameend} \& \bibinfo{author}{M.~\surnamestart
  Daley\surnameend} (\bibinfo{year}{2007}): \emph{\bibinfo{title}{Regulated
  {RNA} Rewriting: Modelling {RNA} Editing with Guided Insertion}}.
\newblock {\sl \bibinfo{journal}{Theoretical Computer Science}}
  \bibinfo{volume}{387}(\bibinfo{number}{2}), pp. \bibinfo{pages}{103--112},
  \doi{10.1016/j.tcs.2007.07.030}.

\bibitemdeclare{article}{Blum1990189}
\bibitem{Blum1990189}
\bibinfo{author}{B.~\surnamestart Blum\surnameend},
  \bibinfo{author}{N.~\surnamestart Bakalara\surnameend} \&
  \bibinfo{author}{L.~\surnamestart Simpson\surnameend} (\bibinfo{year}{1990}):
  \emph{\bibinfo{title}{A Model for RNA Editing in Kinetoplastid Mitochondria:
  RNA Molecules Transcribed From Maxicircle DNA Provide the Edited
  Information}}.
\newblock {\sl \bibinfo{journal}{Cell}} \bibinfo{volume}{60}, pp.
  \bibinfo{pages}{189--198}, \doi{10.1016/0092-8674(90)90735-W}.

\bibitemdeclare{inproceedings}{FWHH05}
\bibitem{FWHH05}
\bibinfo{author}{Xuezheng \surnamestart Fu\surnameend}, \bibinfo{author}{Hao
  \surnamestart Wang\surnameend}, \bibinfo{author}{W.~\surnamestart
  Harrison\surnameend} \& \bibinfo{author}{R.~\surnamestart
  Harrison\surnameend} (\bibinfo{year}{2005}): \emph{\bibinfo{title}{{RNA}
  Pseudoknot Prediction using Term Rewriting}}.
\newblock In: {\sl \bibinfo{booktitle}{Proc.\ BIBE'05, Minneapolis}},
  \bibinfo{publisher}{IEEE Computer Society}, pp. \bibinfo{pages}{169--176},
  \doi{10.1109/BIBE.2005.50}.

\bibitemdeclare{article}{Hea87:bmb}
\bibitem{Hea87:bmb}
\bibinfo{author}{T.~\surnamestart Head\surnameend} (\bibinfo{year}{1987}):
  \emph{\bibinfo{title}{Formal Language Theory and {DNA}: An Analysis of the
  Generative Capacity of Specific Recombinant Behaviors}}.
\newblock {\sl \bibinfo{journal}{Bulletin of Mathematical Biology}}
  \bibinfo{volume}{49}(\bibinfo{number}{6}), pp. \bibinfo{pages}{737--759},
  \doi{10.1016/S0092-8240(87)90018-8}.

\bibitemdeclare{article}{HW04}
\bibitem{HW04}
\bibinfo{author}{D.~\surnamestart Hofbauer\surnameend} \&
  \bibinfo{author}{J.~\surnamestart Waldmann\surnameend}
  (\bibinfo{year}{2004}): \emph{\bibinfo{title}{Deleting String Rewriting
  Systems Preserve Regularity}}.
\newblock {\sl \bibinfo{journal}{Theoretical Computer Science}}
  \bibinfo{volume}{327}, pp. \bibinfo{pages}{301--317},
  \doi{10.1016/j.tcs.2004.04.009}.

\bibitemdeclare{book}{HU79}
\bibitem{HU79}
\bibinfo{author}{J.E. \surnamestart Hopcroft\surnameend} \&
  \bibinfo{author}{J.D. \surnamestart Ullman\surnameend}
  (\bibinfo{year}{1979}): \emph{\bibinfo{title}{Introduction to Automata
  Theory, Languages and Computation}}.
\newblock \bibinfo{publisher}{Addison-Wesley}.

\bibitemdeclare{article}{CH91:dam}
\bibitem{CH91:dam}
\bibinfo{author}{K.~Cullik \surnamestart II\surnameend} \&
  \bibinfo{author}{T.~\surnamestart Harju\surnameend} (\bibinfo{year}{1991}):
  \emph{\bibinfo{title}{Splicing Semigroups and Dominoes and {DNA}}}.
\newblock {\sl \bibinfo{journal}{Discrete Applied Mathematics}}
  \bibinfo{volume}{31}(\bibinfo{number}{3}), pp. \bibinfo{pages}{261--271},
  \doi{10.1016/0166-218X(91)90054-Z}.

\bibitemdeclare{inproceedings}{L08}
\bibitem{L08}
\bibinfo{author}{P.~\surnamestart Leupold\surnameend} (\bibinfo{year}{2008}):
  \emph{\bibinfo{title}{On Regularity-Preservation by String-Rewriting
  Systems}}.
\newblock In \bibinfo{editor}{C.~\surnamestart Mart\'{\i}n-Vide\surnameend},
  \bibinfo{editor}{F.~\surnamestart Otto\surnameend} \&
  \bibinfo{editor}{H.~\surnamestart Fernau\surnameend}, editors: {\sl
  \bibinfo{booktitle}{Proc.\ LATA 2008}}, \bibinfo{publisher}{LNCS 5196}, pp.
  \bibinfo{pages}{345--356}, \doi{10.1007/978-3-540-88282-4\_32}.

\bibitemdeclare{article}{MPRV05}
\bibitem{MPRV05}
\bibinfo{author}{M.~\surnamestart Margenstern\surnameend},
  \bibinfo{author}{G.~\surnamestart Paun\surnameend},
  \bibinfo{author}{Y.~\surnamestart Rogozhin\surnameend} \&
  \bibinfo{author}{S.~\surnamestart Verlan\surnameend} (\bibinfo{year}{2005}):
  \emph{\bibinfo{title}{Context-free Insertion-deletion Systems}}.
\newblock {\sl \bibinfo{journal}{Theoretical Computer Science}}
  \bibinfo{volume}{330}, pp. \bibinfo{pages}{339--348},
  \doi{10.1016/j.tcs.2004.06.031}.

\bibitemdeclare{article}{Pix96:dam}
\bibitem{Pix96:dam}
\bibinfo{author}{D.~\surnamestart Pixton\surnameend} (\bibinfo{year}{1996}):
  \emph{\bibinfo{title}{Regularity of Splicing Languages}}.
\newblock {\sl \bibinfo{journal}{Discrete Applied Mathematics}}
  \bibinfo{volume}{70}(\bibinfo{number}{1}), pp. \bibinfo{pages}{57--79},
  \doi{10.1016/0166-218X(95)00079-7}.

\bibitemdeclare{article}{vdSpekEtAl91}
\bibitem{vdSpekEtAl91}
\bibinfo{author}{H.~van~der \surnamestart Spek\surnameend},
  \bibinfo{author}{G.J. \surnamestart Arts\surnameend}, \bibinfo{author}{R.R.
  \surnamestart Zwaal\surnameend}, \bibinfo{author}{J.~van~den \surnamestart
  Burg\surnameend}, \bibinfo{author}{P.~\surnamestart Sloof\surnameend} \&
  \bibinfo{author}{R.~\surnamestart Benne\surnameend} (\bibinfo{year}{1991}):
  \emph{\bibinfo{title}{Conserved Genes Encode Guide RNAs in Mitochondria of
  \emph{Crithidia Fasciculata}}}.
\newblock {\sl \bibinfo{journal}{The EMBO Journal}}
  \bibinfo{volume}{10}(\bibinfo{number}{5}), pp. \bibinfo{pages}{1217--1224}.

\bibitemdeclare{article}{StuartEtAl97}
\bibitem{StuartEtAl97}
\bibinfo{author}{K.~\surnamestart Stuart\surnameend}, \bibinfo{author}{T.E.
  \surnamestart Allen\surnameend}, \bibinfo{author}{S.~\surnamestart
  Heidmann\surnameend} \& \bibinfo{author}{S.D. \surnamestart
  Seiwert\surnameend} (\bibinfo{year}{1997}): \emph{\bibinfo{title}{RNA editing
  in kinteoplastid protozoa}}.
\newblock {\sl \bibinfo{journal}{Micorbiology and Molecular Biology Reviews}}
  \bibinfo{volume}{61}(\bibinfo{number}{1}), pp. \bibinfo{pages}{105--120}.

\bibitemdeclare{article}{TY03}
\bibitem{TY03}
\bibinfo{author}{A.~\surnamestart Takahara\surnameend} \&
  \bibinfo{author}{T.~\surnamestart Yokomori\surnameend}
  (\bibinfo{year}{2003}): \emph{\bibinfo{title}{On the Computational Power of
  Insertion-Deletion Systems}}.
\newblock {\sl \bibinfo{journal}{Natural Computing}}
  \bibinfo{volume}{2}(\bibinfo{number}{4}), pp. \bibinfo{pages}{321--336},
  \doi{10.1023/B:NACO.0000006769.27984.23}.

\bibitemdeclare{inproceedings}{Zan10:lata}
\bibitem{Zan10:lata}
\bibinfo{author}{H.~\surnamestart Zantema\surnameend} (\bibinfo{year}{2010}):
  \emph{\bibinfo{title}{Complexity of Guided Insertion-Deletion in
  {RNA}-Editing}}.
\newblock In \bibinfo{editor}{A.-H. \surnamestart Dediu\surnameend},
  \bibinfo{editor}{H.~\surnamestart Fernau\surnameend} \&
  \bibinfo{editor}{C.~\surnamestart Mart\'{\i}n-Vide\surnameend}, editors: {\sl
  \bibinfo{booktitle}{Proc.\ LATA 2010}}, \bibinfo{publisher}{LNCS 6031}, pp.
  \bibinfo{pages}{608--619}, \doi{10.1007/978-3-642-13089-2\_51}.

\end{thebibliography}

\end{document}